\documentclass[11pt]{article}
 \pdfoutput=1                           
\usepackage{a4wide}
\usepackage{amsmath,amsfonts,amssymb,amsthm,amscd}
\usepackage{graphicx}

\usepackage{color}
\definecolor{modra}{rgb}{0,0,.8}
\definecolor{piros}{rgb}{.8,0,0}
\usepackage{enumerate}

\usepackage{hyperref}

\if10     
\usepackage[mathlines]{lineno}
\newcommand*\patchAmsMathEnvironmentForLineno[1]{%
  \expandafter\let\csname old#1\expandafter\endcsname\csname #1\endcsname
  \expandafter\let\csname oldend#1\expandafter\endcsname\csname end#1\endcsname
  \renewenvironment{#1}%
     {\linenomath\csname old#1\endcsname}%
     {\csname oldend#1\endcsname\endlinenomath}}%
\newcommand*\patchBothAmsMathEnvironmentsForLineno[1]{%
  \patchAmsMathEnvironmentForLineno{#1}%
  \patchAmsMathEnvironmentForLineno{#1*}}%
\AtBeginDocument{%
\patchBothAmsMathEnvironmentsForLineno{equation}%
\patchBothAmsMathEnvironmentsForLineno{align}%
\patchBothAmsMathEnvironmentsForLineno{flalign}%
\patchBothAmsMathEnvironmentsForLineno{alignat}%
\patchBothAmsMathEnvironmentsForLineno{gather}%
\patchBothAmsMathEnvironmentsForLineno{multline}%
}
\linenumbers
\fi

\definecolor{modra3}{rgb}{.1,.0,.4}
\hypersetup{colorlinks=true, linkcolor=modra3, urlcolor=modra3, citecolor=modra3, pdfpagemode=UseNone, pdfstartview=} 

      \theoremstyle{plain}

      \newtheorem{theorem}{Theorem} 
      \newtheorem{lemma}[theorem]{Lemma}

      \newtheorem{observation}[theorem]{Observation}

\newtheorem{conjecture_}{Conjecture}

      \theoremstyle{definition}

      \theoremstyle{remark}

\def\marrow{{\marginpar[\hfill$\Rrightarrow$]{$\Lleftarrow$}}}

\def\jk#1{\color{modra} {\sc JK: }{\marrow\sf #1} \normalcolor}

\def\cNP{\hbox{\rm \sffamily NP}}

\def\inst#1{$^{#1}$}

\begin{document}

\title{Clustered planarity testing revisited}

\author{
Radoslav Fulek\inst{1}$^{,}$\inst{2}\thanks{The author gratefully acknowledges support from the Swiss National Science Foundation Grant No.  200021-125287/1 and PBELP2\_146705, and ESF Eurogiga project GraDR as GA\v{C}R GIG/11/E023.}
  \and
Jan Kyn\v{c}l\inst{1}$^{,}$\inst{3}\thanks{Supported by Swiss National Science Foundation Grants 200021-137574 and 200020-14453, by the ESF Eurogiga project GraDR as GA\v{C}R GIG/11/E023, by the grant no. 14-14179S of the Czech Science Foundation (GA\v{C}R) and by the grant SVV-2013-267313 (Discrete Models and Algorithms).}
  \and
Igor Malinovi\'c\inst{4}
  \and
D\"om\"ot\"or P\'alv\"olgyi\inst{5}\thanks{Supported by Hungarian National Science Fund (OTKA), under grant PD 104386 and under grant NN 102029 (EUROGIGA project GraDR 10-EuroGIGA-OP-003) and the J\'anos Bolyai Research Scholarship of the Hungarian Academy of Sciences.} }

\date{}

\maketitle

\begin{center}
{\footnotesize
\inst{1}
Department of Applied Mathematics and Institute for Theoretical Computer Science, \\
Charles University, Faculty of Mathematics and Physics, \\
Malostransk\'e n\'am.~25, 118 00~ Prague, Czech Republic; \\
\texttt{kyncl@kam.mff.cuni.cz, radoslav.fulek@gmail.com}
\\\ \\
\inst{2}
 Department of Industrial Engineering and Operations Research,
 Columbia University, \\
 New York City, NY, USA
\\\ \\
\inst{3}
\'Ecole Polytechnique F\'ed\'erale de Lausanne, Chair of Combinatorial Geometry, \\
EPFL-SB-MATHGEOM-DCG, Station 8, CH-1015 Lausanne, Switzerland
\\\ \\
\inst{4}
\'Ecole Polytechnique F\'ed\'erale de Lausanne, Chair of Discrete Optimization, \\
EPFL-SB-MATHAA-DISOPT, Station 8, CH-1015 Lausanne, Switzerland; \\
\texttt{igor.malinovic@epfl.ch}
\\\ \\
\inst{5}
Institute of Mathematics, \\
 E\"otv\"os University,
  P\'azm\'any P\'eter s\'et\'any 1/C,
 H-1117 Budapest, Hungary; \\
\texttt{domotorp@gmail.com}

}
\end{center}


\begin{abstract}
The Hanani--Tutte theorem is a classical result proved for the first time in the 1930s that characterizes planar graphs
as graphs that admit a drawing in the plane in which every pair of edges not sharing a vertex cross an
even number of times.
We generalize this result to clustered graphs with two disjoint clusters, and show
that a straightforward extension to flat clustered graphs with three or more disjoint clusters is not possible.
For general clustered graphs we show a variant of the  Hanani--Tutte theorem in the case when
each cluster induces a connected subgraph.

Di Battista and Frati proved that clustered planarity of embedded clustered graphs whose every face is incident to at most five vertices can be tested in polynomial time. We give a new and short proof of this result, using the matroid intersection algorithm.
\end{abstract}

\section{Introduction}
\label{section_intro}

Investigation of graph planarity can be traced back to the 1930s and developments accomplished at that time by
 Hanani~\cite{Ha34_uber}, Kuratowski~\cite{Kuratowski30_sur}, Whitney~\cite{Whitney32_nonseparable} and others.
 Forty years later, with the advent of computing, a linear-time algorithm for graph planarity was discovered~\cite{HoTa74_planarity}.
Nowadays, a polynomial-time algorithm for testing whether a graph admits a crossing-free drawing in the plane
could almost be considered a folklore result.

Nevertheless,  many variants of planarity are still only poorly understood. As a consequence of this state of affairs,
the corresponding decision problem for these variants has neither been shown to be polynomial nor \cNP-hard.
\emph{Clustered
planarity} is one of the most prominent~\cite{CoDiBa05_clustered} of such planarity notions. Roughly speaking, an instance of this problem is a graph whose vertices are partitioned into clusters.
The question is whether the graph can be drawn in the plane so that the vertices in the same cluster belong to the same simple closed region and no edge crosses the boundary of a particular region more than once.
The aim of the present work is to offer novel perspectives on clustered planarity, which seem to be worth pursuing in order to improve our understanding of the problem.

More precisely, a \emph{clustered graph} is a pair $(G,T)$ where $G=(V,E)$ is a graph and $T$ is a rooted tree whose set of leaves is the set of vertices of $G$.
The non-leaf vertices of $T$ represent the clusters. Let $C(T)$ be the set of non-leaf vertices of $T$.
For each $\nu\in C(T)$, let $T_{\nu}$ denote the subtree of $T$ rooted at $\nu$. The {\em cluster\/} $V(\nu)$ is the set of leaves of $T_{\nu}$. 
A clustered graph $(G,T)$ is \emph{flat} if all non-root clusters are children of the root cluster; that is, if every root-leaf path in $T$ has at most three vertices. When discussing flat clustered graphs, which is basically everywhere except Sections~\ref{section_intro},~\ref{section_alg} and~\ref{section_c_connected}, by ``cluster'' we will refer only to the non-root clusters.

A \emph{drawing} of $G$ is a representation of $G$ in the plane where every vertex is represented by a unique point and every
edge $e=uv$ is represented by a simple arc joining the two points that represent $u$ and $v$. If it leads to no confusion, we do not distinguish between
a vertex or an edge and its representation in the drawing and we use the words ``vertex'' and ``edge'' in both contexts. We assume that in a drawing no edge passes through a vertex,
no two edges touch and every pair of edges cross in finitely many points.
We assume that the above properties of a drawing of $G$ are maintained during any continuous deformation of the drawing of $G$
except for intermediate  one-time events when two edges touch in a single point or an edge passes through a vertex.

 A drawing of a graph is an \emph{embedding} if no two edges cross.

A clustered graph $(G,T)$ is \emph{clustered planar} (or briefly \emph{c-planar}) if $G$ has an embedding in the plane such that

\begin{enumerate}[(i)]
\item for every $\nu\in C(T)$, there is a topological disc $\Delta(\nu)$ containing all the leaves of $T_{\nu}$ and no other vertices of $G$,

\item if $\mu\in T_{\nu}$, then $\Delta(\mu) \subseteq \Delta(\nu)$,

\item if $\mu_1$ and $\mu_2$ are children of $\nu$ in $T$, then $\Delta(\mu_1)$ and $\Delta(\mu_2)$ are internally disjoint, and

\item for every $\nu\in C(T)$, every edge of $G$ intersects the boundary of the disc $\Delta(\nu)$ at most once.
\end{enumerate}
A \emph{clustered drawing (or embedding)} of a clustered graph $(G,T)$ is a drawing (or embedding, respectively) of $G$ satisfying (i)--(iv).
See Figures~\ref{fig_1_1_example_definition_clustered} and~\ref{fig_1_2_example_non_c_planar} for an illustration.
We will be using the word ``cluster'' for both the topological disc $\Delta(\nu)$ and the subset of vertices $V(\nu)$.

\begin{figure}
\begin{center}
 \ifpdf \includegraphics[scale=1]{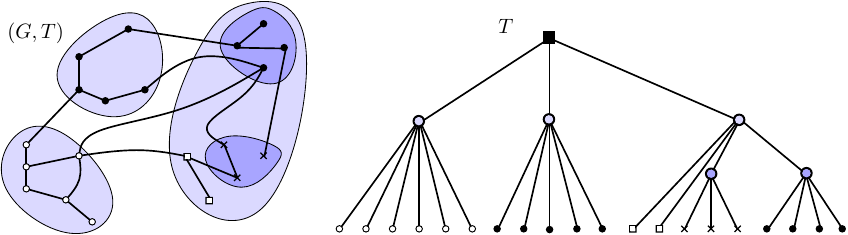} \fi
\end{center}
 \caption{A clustered embedding of a clustered graph $(G,T)$ and its tree $T$.}
 \label{fig_1_1_example_definition_clustered}
\end{figure}

\begin{figure}
\begin{center}
 \ifpdf \includegraphics[scale=1]{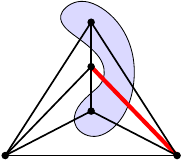} \fi
\end{center}
 \caption{A clustered graph with one non-root cluster, which is not c-planar.}
 \label{fig_1_2_example_non_c_planar}
\end{figure}

\paragraph{A brief history of clustered planarity.}
The notion of clustered planarity was introduced by Feng, Cohen and Eades~\cite{FCE95_how,FCE95b_planarity} under the name {\em c-planarity}. A similar problem, {\em hierarchical planarity}, was considered already by Lengauer~\cite{Lengauer89_hierarchical}.
Since then an efficient algorithm
for c-planarity testing or embedding has been discovered only in some special cases. The general problem whether the c-planarity of a clustered graph $(G,T)$ can be tested in polynomial time is wide open, already when we restrict ourselves to three pairwise disjoint clusters and the case when the combinatorial embedding of $G$ is a part of the input!

A clustered graph $(G,T)$  is \emph{c-connected} if every cluster of $(G,T)$ induces a connected subgraph. See Figure~\ref{fig_5_1_c-connected1}.
In order to test a c-connected clustered graph $(G,T)$ for c-planarity, it is enough to test whether there exists an embedding of $G$ such that for every $\nu\in C(T)$, all vertices of $V(G) \setminus V(\nu)$ are drawn in the outer face of the subgraph induced by $V(\nu)$~\cite[Theorem 1]{FCE95b_planarity}.
Cortese et al.~\cite{CDiBFPP08_c_connected} gave a structural characterization of c-planarity for c-connected clustered graphs and provided a linear-time algorithm.
Gutwenger et al.~\cite{GJLMPW02_advances} constructed a polynomial algorithm for a more general case of \emph{almost connected} clustered graphs, which can be also used
for the case of flat clustered graphs with two clusters forming a partition of the vertex set. Biedl~\cite{Biedl98_III} gave the first polynomial-time algorithm for c-planarity with two clusters, including the case of straight-line or $y$-monotone drawings. An alternative approach to the case of two clusters was given by Hong and Nagamochi~\cite{Hong09_twopage}.
  On the other hand, only very little is known in the case of three clusters, where the only clustered graphs for which a polynomial algorithm for c-planarity is known are clustered cycles~\cite{CDibPP05_cycles}.

\paragraph{Hanani--Tutte theorem.}
The Hanani--Tutte theorem~\cite{Ha34_uber,Tutte70_toward} is a classical result that provides an algebraic characterization of planarity with interesting theoretical and algorithmic consequences; see Section~\ref{section_alg}.
The (strong) Hanani--Tutte theorem says that a graph is planar if it can be drawn in the plane so that no pair of independent edges crosses an odd number of times.
Moreover, its variant known as the weak Hanani--Tutte theorem~\cite{CN00_thrackles,PT00_which,PSS06_removing} states that if $G$ has a drawing $\mathcal{D}$ where every pair of edges cross an even number of times, then $G$ has an embedding that preserves the cyclic order of edges at vertices in $\mathcal{D}$.
Note that the weak variant does not directly follow from the strong Hanani--Tutte theorem.
For sub-cubic graphs, the weak variant implies the strong variant.
Other variants of the Hanani--Tutte theorem were proved for surfaces of higher genus~\cite{PSS09_projective,PSS09b_surfaces},  $x$-monotone drawings~\cite{FPSS12_monotone,PT04_monotone},
partially embedded planar graphs, and several special cases of simultaneously embedded planar graphs~\cite{Sch13b_towards}.
See~\cite{Sch13_hananitutte} for a recent survey on applications of the Hanani--Tutte theorem and related results.

We prove a variant of the (strong) Hanani--Tutte theorem for flat clustered graphs with two
clusters forming a partition of the vertex set. Similarly to other variants of the Hanani--Tutte theorem, as a byproduct of our result, we immediately obtain a polynomial-time algorithm for testing c-planarity in this special case. The algorithm essentially consists of solving a linear system of equations over $\mathbb{Z}_2$.
The running time of the algorithm is in $O(|V(G)|^{2\omega})$,  where $O(n^{\omega})$ is the complexity of multiplication of two square $n\times n$ matrices; see Section~\ref{section_alg}.
The best current algorithms for matrix multiplication give $\omega<2.3729$~\cite{LeGall14_powers,Williams12_faster}.
Since our linear system is sparse, it is also possible to use Wiedemann's randomized algorithm~\cite{Wie86_sparse}, with expected running time $O(n^4\log{n}^2)$ in our case.

Although the worst-case running time of our algorithm is not competitive, we believe this does not make our results less interesting, since the purpose of our direction of research lies more in theoretical foundations than in its immediate consequences. Moreover, the worst-case running time analysis often gives an unfair perspective
on the performance of algebraic algorithms, such as the simplex method.

We remark that there exist more efficient algorithms for planarity testing using the Hanani--Tutte theorem such as those in~\cite{FMR06_tremaux,dFR85_char_tremaux}, which run in linear time;
 see also \cite[Section 1.4.1]{Sch13_hananitutte}. Moreover, in the case of $x$-monotone drawings a computational study~\cite{CZ14_upward} showed that the Hanani--Tutte approach~\cite{FPSS12_monotone} performs really well in practice. This should come as no surprise, since Hanani--Tutte theory seems to provide solid theoretical foundations for graph planarity
that bring together its combinatorial, algebraic, and computational aspects~\cite{Sch13b_towards}.

\paragraph{Notation.}
In this paper we assume that $G=(V,E)$ is a graph, and we state all our theorems for graphs. However, in some of our proofs we also use multigraphs, that is, generalized graphs that can have multiple edges and multiple loops. Most of the notions defined for graphs extend naturally to multigraphs, and thus we use them without generalizing them explicitly.
We use a shorthand notation $G-v$ for $(V\setminus \{v\},E\setminus \{vw| \  vw \in E\})$, and $G\cup E'$ for $(V, E\cup E')$.
The \emph{rotation} at a vertex $v$ is the clockwise cyclic order of the end pieces of edges incident to $v$. The \emph{rotation system} of a graph is the set of rotations at all its vertices.
We say that two embeddings of a graph are the \emph{same} if they have the same rotation system up to switching the orientations of all the rotations simultaneously.
We say that two edges in a graph are \emph{independent} if they do not share a vertex.
An edge in a drawing is {\em even\/} if it crosses every other edge an even number of times.
A drawing of a graph is \emph{even} if all edges are even.
A drawing of a graph is \emph{independently even} if every pair of independent edges in the drawing cross an even number of times.

\paragraph{Hanani--Tutte theorem for clustered graphs.}
A clustered graph $(G,T)$ is {\em two-clustered\/} if the root of $T$ has exactly two children, $A$ and $B$, and every vertex of $G$ is a child of either $A$ or $B$ in $T$. In other words, $A$ and $B$ are the only non-root clusters and they form a partition of the vertex set of $G$. Obviously, two-clustered graphs form a subclass of flat clustered graphs.
We extend both the weak and the strong variant of the Hanani--Tutte theorem to two-clustered graphs.

\begin{theorem}
\label{thm:WeakClusterHT}
If a two-clustered graph $(G,T)$ admits an even clustered drawing $\mathcal{D}$ in the plane then $(G,T)$ is c-planar.
Moreover, $(G,T)$ has a clustered embedding with the same rotation system as $\mathcal{D}$.
\end{theorem}

Theorem~\ref{thm:WeakClusterHT} has been recently generalized by the first author to the case of strip planarity~\cite{Fu14+_towards}.

\begin{theorem}
\label{thm:StrongClusterHT}
If a two-clustered graph $(G,T)$ admits an independently even clustered drawing in the plane then $(G,T)$ is c-planar.
\end{theorem}

We also prove a strong Hanani--Tutte theorem for c-connected clustered graphs.

\begin{theorem}
\label{thm:StrongClusterHTC}
If a c-connected clustered graph $(G,T)$ admits an independently even clustered drawing in the plane then $(G,T)$ is c-planar.
\end{theorem}

On the other hand, we exhibit examples of clustered graphs with more than two disjoint clusters that are not c-planar, but admit an even clustered drawing. This shows that a straightforward extension of Theorem~\ref{thm:WeakClusterHT} and Theorem~\ref{thm:StrongClusterHT} to flat clustered graphs with more than two clusters is not possible.

\begin{theorem}
\label{thm:counter}
For every $k\ge 3$ there exists a flat clustered cycle with $k$ clusters that is not c-planar but
admits an even clustered drawing in the plane.
\end{theorem}

 Gutwenger, Mutzel and Schaefer~\cite{GMS14_practical} recently showed that by using the reduction from~\cite{Sch13b_towards} our counterexamples can be turned into counterexamples for~\cite[Conjecture 1.2]{Sch13b_towards}\footnote{For a graph $G$ drawn in the plane the conjecture claims that by redrawing $G$ we can eliminate crossings in a subgraph $H$ of $G$ consisting of independently even edges without introducing new pairs of non-adjacent edges crossing an odd number of times.} and for a variant of the Hanani--Tutte theorem for two simultaneously embedded planar graphs~\cite[Conjecture 6.20]{Sch13b_towards}.

\paragraph{Embedded clustered graphs with small faces.}
A pair $(\mathcal{D}(G),T)$ is an \emph{embedded clustered graph} if $(G,T)$ is a clustered graph and $\mathcal{D}(G)$ is an embedding of $G$ in the plane, not necessarily a clustered embedding.
The embedded clustered graph $(\mathcal{D}(G),T)$ is \emph{c-planar} if it can be extended to a
clustered embedding of $(G,T)$ by choosing a topological disc for each cluster.

We give an alternative polynomial-time algorithm for deciding c-planarity of embedded flat clustered graphs with small faces, reproving a result of Di Battista and Frati~\cite{DiBF09_small_faces}. Our algorithm is based on the matroid intersection theorem. Its running time is $O(|V(G)|^{3.5})$ by~\cite{Cu86_matroid}, so it
does not outperform the linear algorithm from~\cite{DiBF09_small_faces}. Similarly as for our other results, we see its purpose more in mathematical foundations than in giving an efficient algorithm.
 We find it quite surprising that by using completely different techniques we obtained an algorithm
for exactly the same case. Our approach is very similar to a technique used by Katz, Rutter and Woeginger~\cite{KRW12_switch} for deciding the global
connectivity of switch graphs.

\begin{theorem}\label{thm:matroid}{\rm\cite{DiBF09_small_faces}}
Let $\mathcal{D}(G)$ be an embedding of a graph $G$ in the plane such that all its faces are incident to at most five vertices.
Let $(G,T)$ be a flat clustered graph.
The problem whether $(G,T)$ admits a c-planar embedding in which $G$ keeps its embedding $\mathcal{D}(G)$ can be solved in polynomial time.
\end{theorem}

\paragraph{Organization.}
The rest of the paper is organized as follows. In Section~\ref{section_alg} we describe an algorithm for
c-planarity testing based on Theorem~\ref{thm:StrongClusterHT}.
  In Section~\ref{section_weak} we prove Theorem~\ref{thm:WeakClusterHT}.
  In Section~\ref{section_strong} we prove Theorem~\ref{thm:StrongClusterHT}.
  In Section~\ref{section_c_connected} we prove Theorem~\ref{thm:StrongClusterHTC}.
In Section~\ref{section_3clusters} we provide a family of counterexamples to the variant of the Hanani--Tutte theorem for clustered graphs with three clusters, and discuss properties that every such counterexample, whose underlying abstract graph is a cycle, must satisfy.
In Section~\ref{section_faces} we prove Theorem~\ref{thm:matroid}.
We conclude with some remarks in Section~\ref{section_epilogue}.

\section{Algorithm}
\label{section_alg}

Let $(G, T)$ be a clustered graph for which the corresponding variant of the strong Hanani--Tutte theorem holds, that is,
the existence of an independently even clustered drawing of $(G,T)$ implies that $(G,T)$ is c-planar.

Our algorithm for c-planarity testing is an adaptation of the algorithm for planarity testing from~\cite[Section~1.4.2]{Sch13_hananitutte}.
The algorithm starts with an arbitrary clustered drawing $\mathcal{D}$ of $(G,T)$. Such a drawing always exists: for example, we can traverse the tree $T$ using depth-first search and place the vertices of $G$ on a circle in the order encountered during the search. Then we draw every edge as a straight-line segment. Since every cluster consists of consecutive vertices on the circle, the topological discs representing the clusters can be drawn easily.
The algorithm tests whether the edges of the initial drawing $\mathcal{D}$ can be continuously deformed to form
an independently even clustered drawing $\mathcal{D}_0$ of $(G,T)$. This is done by constructing and solving a system of linear equations over $\mathbb{Z}_2$. By the corresponding variant of the strong Hanani--Tutte theorem, the existence of such a drawing $\mathcal{D}_0$ is equivalent to the c-planarity of $(G,T)$.

Now we describe the algorithm in more details. We start with the original algorithm for planarity testing and then show how to modify it for c-planarity testing.

During a ``generic'' continuous deformation from $\mathcal{D}$ to some other drawing $\mathcal{D}'$, the parity of the number of crossings between a pair of independent edges is affected only when an edge $e$ passes over a vertex $v$ that is not incident to $e$,
in which case we change the parity of the number of crossings of $e$ with all the edges incident to $v$; see Figure~\ref{obr_2_1_switch}. We call such an event an \emph{edge-vertex switch}. Note that every edge-vertex switch can be performed independently of others, for any initial drawing: we can always deform a given edge $e$ to pass close to the given vertex $v$, while introducing new crossings with every edge ``far from $v$'' only in pairs; that is, after every event when $e$ touches another edge, a pair of new crossings is created.
 For our purpose the deformation from $\mathcal{D}$ to $\mathcal{D}'$ can be represented by the set of edge-vertex switches that were performed an odd number of times during the deformation.
  An edge-vertex switch of an edge $e$ with a vertex $v$ is denoted by the ordered pair $(e,v)$.

\begin{figure}
\begin{center}
 \ifpdf\includegraphics[scale=1]{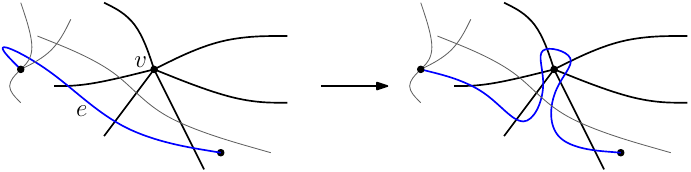}\fi
\end{center}
 \caption{A continuous deformation of $e$ resulting in an edge-vertex switch $(e,v)$.}
 \label{obr_2_1_switch}
\end{figure}

A drawing of $(G,T)$ can then be represented as a vector ${\bf v}\in\mathbb{Z}_2^{M}$, where
$M$ is the number of unordered pairs of independent edges. The component of ${\bf v}$ corresponding to a pair $\{e,f\}$ is $1$ if $e$ and $f$ cross an odd number of times and $0$ otherwise.
Let $e$ be an edge of $G$ and $v$ a vertex of $G$ such that $v \notin e$. Performing an edge-vertex switch $(e,v)$ corresponds to adding the vector ${\bf w}_{(e,v)}\in\mathbb{Z}_2^{M}$ whose only components equal to $1$ are those indexed by pairs $\{e,f\}$ where $f$ is incident to $v$.
The set of all drawings of $G$ that can be obtained from $\mathcal{D}$ by edge-vertex switches then corresponds to an affine subspace ${\bf v}+W$, where $W$ is the subspace generated by the set $\{{\bf w}_{(e,v)}; v \notin e\}$.
The algorithm tests whether ${\bf 0}\in {\bf v}+W$, which is equivalent to the solvability of a system of linear equations over $\mathbb{Z}_2$.

The difference between the original algorithm for planarity testing and our version for $c$-planarity testing is the following.
To keep the drawing of $(G,T)$ clustered after every deformation, for every edge $e=v_1v_2$, we allow only those edge-vertex switches $(e,v)$ such that $v$ is a child of some vertex of the shortest path between $v_1$ and $v_2$ in $T$. Such vertices $v$ are precisely those that are not separated from $e$ by cluster boundaries.

We also include \emph{edge-cluster switches} $(e,C)$ where $C$ is a child of some vertex of the shortest path between $v_1$ and $v_2$ in $T$. An edge-cluster switch $(e,C)$ moves $e$ over the whole topological disc representing $C$; see Figure~\ref{obr_2_2_switch_in_clustered}. Combinatorially, this is equivalent to performing all the edge-vertex switches $(e,v), v\in C$, simultaneously.
The corresponding vector ${\bf w}_{(e,C)}$ is the sum of all ${\bf w}_{(e,v)}$ for $v\in C$. Therefore, the set of allowed switches generates a subspace $W_c$ of $W$. Since every allowed switch can be performed in every clustered drawing, every vector from $W_c$ can be realized by some continuous deformation. Moreover, every clustered drawing of $(G,T)$ can be obtained from any other clustered drawing of $(G,T)$ by a homeomorphism of the plane and by a sequence of finitely many continuous deformations of the edges, where each of the deformations can be represented by a subset of allowed switches.
Indeed, by~\cite[Theorem 1.18]{Sch13_hananitutte} or by the discussion of the original algorithm in previous paragraphs, the vectors $\mathbf{v}$ and $\mathbf{v'}$ corresponding to two clustered drawings $\mathcal{D}$ and $\mathcal{D}'$ of $(G,T)$ differ by a vector $\mathbf{w}\in W$. We claim that $\mathbf{w}\in W_c$. Suppose that $\mathcal{D}$ and $\mathcal{D}'$ have the same vertices. Let $e$ be an edge of $G$, let $e_0$ be the curve representing $e$ in $\mathcal{D}$, and let $e_1$ be the curve representing $e$ in $\mathcal{D}'$. Let $\gamma$ be the closed curve obtained by joining $e_0$ and $e_1$. Let $S$ be the set of vertices of $G$ ``inside'' $\gamma$; see Subsection~\ref{sub_proof_of_theorem2} part 2) for the definition. For every cluster $C$ that $e$ cannot cross, all the vertices of $C$ belong to the same connected region of $\mathbb{R}^2\setminus \gamma$; in particular, they are all ``inside'' or all ``outside'' $\gamma$. For every cluster $C$ whose vertices are ``inside'' $\gamma$, we perform the switch $(e,C)$ and perform the corresponding deformation on the curve $e_0$. Let $e_{1/2}$ be the resulting curve. The closed curve obtained by joining $e_{1/2}$ and $e_1$ has all the vertices of $G$ ``outside''. Therefore, if we now deform $e_{1/2}$ to $e_1$ arbitrarily, every vertex will be crossed an even number of times, so no changes in the parity of crossings between independent edges will occur.


Our algorithm then tests whether ${\bf 0}\in {\bf v}+W_c$.

\begin{figure}
\begin{center}
 \ifpdf\includegraphics[scale=1]{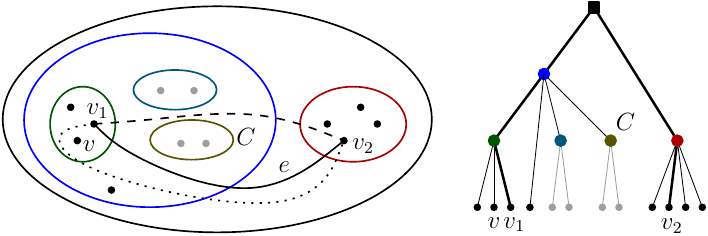}\fi
\end{center}
 \caption{Left: an edge-vertex switch $(e,v)$ and an edge-cluster switch $(e,C)$. Right: the shortest path between $v_1$ and $v_2$ in $T$. The four light gray vertices in the middle cannot participate in a switch with $e$ individually.}
 \label{obr_2_2_switch_in_clustered}
\end{figure}

Before running the algorithm, we first remove any loops and parallel edges and check whether $|E(G')|<3|V(G')|$ for the resulting graph $G'$. Then we run our algorithm on $(G',T)$. This means solving a system of $O(|E(G')||V(G')|)=O(|V(G)|^2)$ linear equations in $O(|E(G')|^2)=O(|V(G)|^2)$ variables. This can be performed in $O(|V(G)|^{2\omega})\le O(|V(G)|^{4.746})$ time using the algorithm by Ibarra, Moran and Hui~\cite{IMH82_matrix}.

Gutwenger, Mutzel and Schaefer~\cite{GMS14_practical} independently proposed a different algebraic algorithm for testing clustered planarity, based on a reduction to simultaneous planarity. It is not hard to show that their algorithm is equivalent to ours, in the sense that both algorithms accept the same instances of clustered graphs.


\section{Weak Hanani--Tutte for two-clustered graphs}

\label{section_weak}

First, we prove a stronger version of a special case of Theorem~\ref{thm:WeakClusterHT} in which $G$ is a bipartite multigraph with the two parts corresponding to the two clusters. We note that a bipartite multigraph has no loops, but it can have multiple edges. In this stronger version, which is an easy consequence of the weak Hanani--Tutte theorem, we assume only the existence of an arbitrary even drawing of $G$
that does not have to be a clustered drawing.

\begin{lemma}
\label{lemma:WeakBipartiteClusterHT}
Let $(G,T)$ be a two-clustered bipartite multigraph in which the two non-root clusters induce independent sets.
If $G$ admits an even drawing then $(G,T)$ is c-planar.
Moreover, there exists a clustered embedding of $(G,T)$ with the same rotation system as in the given even drawing of $G$.
\end{lemma}

\begin{proof}
We assume that $G=(V,E)$ is connected, since we can draw each connected component separately.
Let $A$ and $B$ be the two clusters of $(G,T)$ forming a partition of $V(G)$.
By the weak Hanani--Tutte theorem~\cite{CN00_thrackles,PSS06_removing} we obtain an embedding $\mathcal{D}$ of $G$ with the same rotation system as in the initial even drawing of $G$.

It remains to show that we can draw the discs representing the clusters. This follows from a much stronger geometric result by Biedl, Kaufmann and Mutzel~\cite[Corollary 1]{BiedlKM98_II}. We need only a weaker, topological, version, which has a very short proof. For each face $f$ of $\mathcal{D}$, we may draw without crossings a set $E_f$ of edges inside $f$ joining one chosen vertex in $A$ incident to $f$ to all other vertices in $A$ incident to $f$. Since the dual graph of $G$ in $\mathcal{D}$ is connected, the multigraph $(A,\bigcup_f E_f)$ is connected as well. Let $E'$ be a subset of $\bigcup_f E_f$ such that $T_A=(A,E')$ is a spanning tree of $A$. A small neighborhood of $T_A$ is an open topological disc $\Delta_A$ containing all vertices of $A$, and the boundary of $\Delta_A$ crosses every edge of $G$ at most once; see Figure~\ref{fig_3_lemma_bipartite}. In the complement of $\Delta_A$ we can easily find a topological disc $\Delta_B$ containing all vertices of $B$, by drawing its boundary partially along the boundary of $\Delta_A$ and partially along the boundary of the outer face of $\mathcal{D}$.
\end{proof}

\begin{figure}
\begin{center}
 \ifpdf\includegraphics[scale=1]{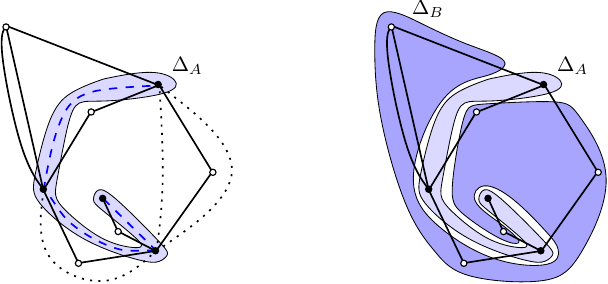}\fi
\end{center}
 \caption{Left: Drawing the disc $\Delta_A$. The edges of $E'$ are dashed, while the edges of $\bigcup_{f}E_{f} \setminus E'$ are dotted. Right: Drawing the disc $\Delta_B$.}
 \label{fig_3_lemma_bipartite}
\end{figure}

\subsection{Proof of Theorem~\ref{thm:WeakClusterHT}}


The proof is inspired by the proof of the weak Hanani--Tutte theorem from~\cite{PSS06_removing}.

Let $A$ and $B$ be the two clusters of $(G,T)$ forming a partition of $V(G)$.
We assume that $G$ is connected, since we can embed each component separately.
We start with an even clustered drawing of $(G,T)$. We proceed by induction on the number of vertices.

\begin{figure}
\begin{center}
 \ifpdf\includegraphics{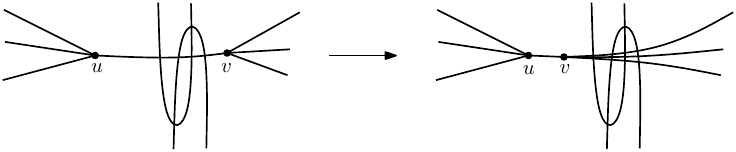}\fi
\end{center}
 \caption{Pulling $v$ towards $u$. The evenness of the drawing is preserved.}
 \label{obr_3_1_pull}
\end{figure}

First, we discuss the inductive step.
If we have an edge $e$ between two vertices $u,v$ in the same part (either $A$ or $B$), we contract $e$ by pulling $v$ along $e$ towards $u$ while dragging all the other edges incident to $v$ along $e$ as well. See Figure~\ref{obr_3_1_pull}. We keep all resulting loops and multiple edges. If some edge crosses itself during the dragging, we eliminate the self-crossing by a local redrawing. The resulting drawing is still a clustered drawing. This operation keeps the drawing even and it also preserves the rotation at each vertex. Then we apply the induction hypothesis and decontract the edge $e$.
This can be done without introducing new crossings, since the rotation system has been preserved during the induction.

In the base step, $G$ is a multigraph consisting of a bipartite multigraph $H$ with parts $A$ and $B$ and possible additional loops at some vertices.
We can embed $H$ by Lemma~\ref{lemma:WeakBipartiteClusterHT}. It remains to embed the loops.
Note that after the contractions, no loop crosses the boundary of a cluster.
Each loop $l$ divides the rotation at its corresponding vertex $v(l)$ into two intervals. One of these intervals contains no end piece of an edge connecting $A$ with $B$, otherwise $l$ would cross some edge of $H$ an odd number of times. Call such an interval a \emph{good} interval in the rotation at $v(l)$.
Observe that there are no two loops $l_1$ and $l_2$ with $v(l_1)=v(l_2)=v$ whose end-pieces have the order $l_1,l_2,l_1,l_2$ in the rotation at $v$, as otherwise the two loops would cross an odd number of times.
Hence, at each vertex the good intervals of every pair of loops are either nested or disjoint.

We use induction on the number of loops to draw all the loops at a given vertex $v$ without crossings and without changing the rotation at $v$. For the inductive step, we remove a loop $l$ whose good interval in the rotation at $v$ is inclusion minimal. Such an interval contains only the two end-pieces of $l$, since
there exist no edges between a pair of vertices in $A$ or $B$.
By induction hypothesis, we can embed the rest of the loops without changing the rotation at $v$. Finally, we can draw $l$ in a close neighborhood of $v$ within the face determined by the original rotation at $v$. This concludes our discussion of the base step of the induction and the proof of the theorem.



\section{Strong Hanani--Tutte for two-clustered graphs}

In this section we prove Theorem~\ref{thm:StrongClusterHT}.
Let $(G,T)$ be a two-clustered graph.
Let $A$ and $B$ be the two clusters of $(G,T)$ forming a partition of $V(G)$.
For a subset $V'\subseteq V(G)$, let $G[V']$ denote the subgraph of $G$ induced by $V'$.
The following lemma gives a characterization of c-planarity for two-clustered graphs, similar to the one for c-connected clustered graphs~\cite[Theorem 1]{FCE95b_planarity}.

\begin{lemma}
\label{lemma:twoClusters}
An embedding of a two-clustered graph $(G,T)$ is a clustered embedding if and only if $G[B]$ is contained in the outer face of $G[A]$ and $G[A]$ is contained in the outer face of $G[B]$.
\end{lemma}

\begin{proof}
The ``only if'' part is trivial. Let $\mathcal{D}$ be an embedding of $G$ in which $G[B]$ is contained in the outer face of $G[A]$ and vice-versa. First we extend $\mathcal{D}$ to an embedding $\mathcal{D}_1$ of a connected two-clustered graph $(G_1,T)$ by adding the minimum necessary number of edges between the components of $\mathcal{D}$. (If $G$ is connected, then $G_1=G$ and $\mathcal{D}_1=\mathcal{D}$.) The embedding $\mathcal{D}_1$ still satisfies the assumptions of the lemma, since adding an edge between two components creates no cycle.

Next we contract each component of $G_1[A]\cup G_1[B]$ in $\mathcal{D}_1$ to a point, while keeping all the loops and multiple edges, and preserving the rotations of the vertices. Let $\mathcal{D}_2$ be the resulting embedding and $(G_2,T_2)$ the corresponding two-clustered multigraph.
The connectedness of $G_2$ and the assumption of the lemma imply that the interior of every loop in $\mathcal{D}_2$ is empty of vertices.
We remove all the loops, and apply Lemma~\ref{lemma:WeakBipartiteClusterHT} to the resulting two-clustered multigraph $(G_3,T_3)$. We obtain topological discs $\Delta_{A}$ and $\Delta_{B}$ certifying the c-planarity of $(G_3,T_3)$. Finally,
we reintroduce the loops and decontract the components of $G_1[A]$ and $G_1[B]$ inside the discs $\Delta_{A}$ and $\Delta_{B}$, respectively. Finally, we delete the edges connecting the components of $G$.
\end{proof}

By the assumption of Theorem~\ref{thm:StrongClusterHT} and the strong Hanani--Tutte theorem, $G$ has an embedding.
However, in this embedding, $G[B]$ does not have to be contained in a single face of $G[A]$ and vice-versa. Hence, we cannot guarantee that a clustered embedding of $(G,T)$ exists so easily.

For an induced subgraph $H$ of $G$, the {\em boundary\/} of $H$ is the set of vertices in $H$ that have a neighbor in $G-H$. We say that an embedding $\mathcal{D}(H)$ of $H$ is {\em exposed\/} if all vertices on the boundary of $H$ are incident to the outer face of $\mathcal{D}(H)$.

The following lemma is an easy consequence of the strong Hanani--Tutte theorem.
It helps us to find an exposed embedding of each connected component $X$ of $G[A] \cup G[B]$.
Later in the proof of Theorem~\ref{thm:StrongClusterHT}
this allows us to remove non-essential parts of each such component $X$ and concentrate only on a subgraph $G'$
of $G$ in which both $G[A]$ and $G[B]$ are outerplanar.

\label{section_strong}

\begin{lemma}
\label{lemma:single_face}
Suppose that $(G,T)$ admits an independently even clustered drawing.
Then every connected component of $G[A] \cup G[B]$ admits an exposed embedding.
\end{lemma}

\begin{proof}
Let $\mathcal{D}$ be an independently even clustered drawing of $(G,T)$. Let $\Delta_A$ and $\Delta_B$ be the two topological discs representing the clusters $A$ and $B$, respectively.

Let $X$ be a component of $G[A]$. (For components of $G[B]$ the proof is analogous.)
Let $\partial X$ be the boundary of $X$.
For $Y\subseteq A$, let $E(Y,B)$ be the set of edges connecting a vertex in $Y$ with a vertex in $B$. Observe that $E(X,B)=E(\partial X, B)$. We replace $B$ by a single vertex $v$ and connect it to all vertices of $\partial X$. We obtain a graph $X'=(V(X)\cup \{v\}, E(X)\cup \{uv; u\in \partial X\}$.

We get an independently even drawing of $X'$ from $\mathcal{D}$ by contracting $\Delta_B$ to a point and removing the vertices in $A\setminus X$ and all parallel edges. By the strong Hanani--Tutte theorem we obtain an embedding of $X'$. By changing this embedding so that $v$ gets to the outer face and then removing $v$ with all incident edges, we obtain an exposed embedding of $X$.
\end{proof}

\subsection{Proof of Theorem~\ref{thm:StrongClusterHT}}
\label{sub_proof_of_theorem2}

The proof is inspired by the proof of the strong Hanani--Tutte theorem from~\cite{PSS06_removing}. Its outline is as follows.
First we obtain a subgraph $G'$ of $G$ containing the boundary of each component of $G[A]$ and $G[B]$ and such that each of $G'[A]$ and $G'[B]$ is a {\em cactus forest}, that is, a graph where every two cycles are edge-disjoint. Equivalently, a cactus forest is a graph with no subdivision of $K_4-e$. A connected component of a cactus forest is called a {\em cactus}.
Then we apply the strong Hanani--Tutte theorem to a graph which is constructed from $G'$ by splitting vertices common to at least two cycles in $G'[A]$ and $G'[B]$, and turning all cycles in $G'[A]$ and $G'[B]$ into wheels.
The wheels guarantee that everything that has been removed from $G$ in order to obtain $G'$ can be inserted back. Finally we draw the clusters using Lemma~\ref{lemma:WeakBipartiteClusterHT}.

Now we describe the proof in detail.
Let $X_1,\ldots, X_k$ be the connected components of $G[A]\cup G[B]$.
By Lemma~\ref{lemma:single_face} we find an exposed embedding $\mathcal{D}(X_i)$ of each $X_i$.
Let $X_i'$ denote the subgraph of $X_i$ obtained by deleting from $X_i$ all the vertices and edges not incident to the outer face of $\mathcal{D}(X_i)$. Observe that $X_i'$ is a cactus.

Let $G'=(\bigcup_{i=1}^k X'_i) \cup E(A,B)$. That is, $G'$ is a subgraph of $G$ that consists of all the cacti $X_i'$ and all edges between the two clusters.
Let $\mathcal{D}'$ denote the drawing of $G'$ obtained from the initial independently even clustered drawing of $G$
by deleting the edges and vertices of $G$ not belonging to $G'$. Thus, $\mathcal{D}'$ is an independently even clustered drawing of $G'$.

In what follows we process the cycles of $G'[A]$ and $G'[B]$ one by one.
We will be modifying $G'$ and also the drawing $\mathcal{D}'$.
We will maintain the property that every processed cycle is vertex-disjoint with all other cycles in $G'[A]$ and $G'[B]$, and every edge of every processed cycle is even in $\mathcal{D}'$. Initially, the property is met as no cycle is processed.
Let $C$ denote an unprocessed cycle in $G'[A]$. For cycles in $G'[B]$, the procedure is analogous. We proceed in several steps.

\paragraph{1) Correcting the rotations.}
For every vertex $v$ of $C$, we redraw the edges incident to $v$ in a small neighborhood of $v$, and change the rotation at $v$, as follows~\cite{PSS06_removing}. If the two edges $e,f$ of $C$ incident to $v$ cross an odd number of times, we redraw one of them, say, $f$, so that they cross evenly. Next, we redraw every other edge incident to $v$ so that it crosses both $e$ and $f$ evenly; see Figure~\ref{obr_4_1_correcting_rotations}. After we perform these modifications at every vertex of $C$, all the edges of $C$ are even. However, some pairs of edges incident to a vertex of $C$ may cross oddly; see Figure~\ref{obr_4_1_correcting_rotations} d). Moreover, no processed cycles have been affected since they are vertex-disjoint with $C$.

\begin{figure}
\begin{center}
 \ifpdf\includegraphics[scale=1]{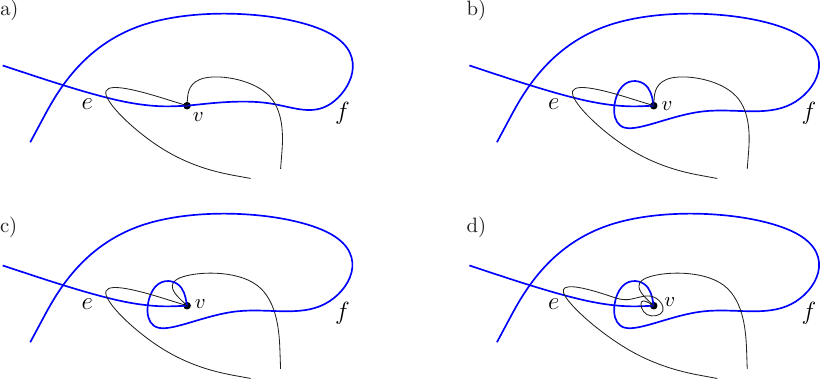}\fi
\end{center}
 \caption{Making $e$ and $f$ even by changing the drawing locally around $v$.}
 \label{obr_4_1_correcting_rotations}
\end{figure}

\paragraph{2) Cleaning the ``inside''.}
We two-color the connected components of $\mathbb{R}^2\setminus C$ so that two regions
sharing a nontrivial part of their boundary receive opposite colors. The existence of such a coloring is a well-known fact; for example, one can color the points of $\mathbb{R}^2\setminus C$ using the parity of the winding number of $C$. We say
that a point not lying on $C$ is ``outside'' of $C$ if it is contained in the region
with the same color as the unbounded region. Otherwise,
such a point is ``inside'' of $C$.

A {\em $C$-bridge\/} in $G'$ is a ``topological'' connected component of $G' - E(C)$; that is, a connected component $K$ of $G' - C$ together with all the edges connecting $K$ with $C$, or a chord of $C$ in $G'$. We say that a $C$-bridge $L$ is {\em outer\/} if all edges of $L$ incident to $C$ attach to the vertices of $C$ from ``outside''. Similarly, we say that a $C$-bridge $L$ is {\em inner\/} if all edges of $L$ incident to $C$ attach to the vertices of $C$ from ``inside''. Since all the edges of $C$ are even, every $C$-bridge is either outer or inner. A {\em $C$-bridge\/} is {\em trivial\/} if it attaches only to one vertex of $C$; otherwise it is {\em nontrivial}. Since $C$ is edge-disjoint with all cycles in $G'[A]$, every nontrivial $C$-bridge contains a vertex of $B$. Since $\mathcal{D}'$ is a clustered drawing of $G'$, all vertices of $G'[B]$ lie ``outside'' of $C$, and so every nontrivial $C$-bridge is outer. Therefore, every inner $C$-bridge is trivial. We redraw every inner $C$-bridge $L$ as follows. Let $v$ be the vertex of $C$ to which $L$ is attached. We select a small region in the neighborhood of $v$ ``outside'' of $C$, and draw $L$ in this region by continuously deforming the original drawing of $L$, so that $L$ crosses no edge outside $L$; see Figure~\ref{obr_4_2_inner_bridges}. After this step, nothing is attached to $C$ from ``inside''.

\begin{figure}
\begin{center}
 \ifpdf\includegraphics[scale=1]{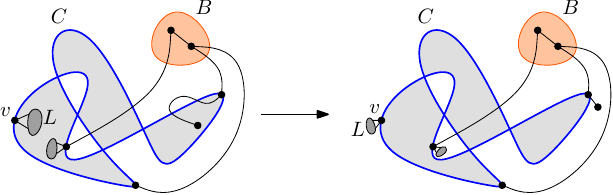}\fi
\end{center}
 \caption{Transforming inner $C$-bridges into outer $C$-bridges. Every nontrivial $C$-bridge contains a vertex in $B$.}
 \label{obr_4_2_inner_bridges}
\end{figure}

\paragraph{3) Vertex splitting.}
Let $v$ be a vertex of $C$ belonging to at least one other cycle in $G'[A]$. Let $x$ and $y$ be the two neighbors of $v$ in $C$. By the previous step, the edges $xv$ and $yv$ are consecutive in the rotation at $v$. We split the vertex $v$ by replacing it with two new vertices $v'$ and $v''$ connected by an edge, and draw them very close to $v$. We replace the edges $xv$ and $yv$ by edges $xv'$ and $yv'$, respectively. For every neighbor $u$ of $v$ that is not on $C$, we replace the edge $uv$ by an edge $uv''$. See Figure~\ref{obr_4_3_vertex_splitting}. Clearly, this vertex-splitting introduces no pair of independent edges crossing oddly. Moreover, after all the splittings, $C$ is vertex-disjoint with all cycles in $G'[A]$.

\begin{figure}
\begin{center}
 \ifpdf\includegraphics[scale=1]{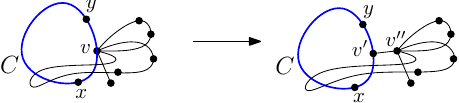}\fi
\end{center}
 \caption{Splitting a vertex $v$ common to several cycles in $G'[A]$.}
 \label{obr_4_3_vertex_splitting}
\end{figure}

\paragraph{4) Attaching the wheels.}
Now we fill the cycle $C$ with a wheel. More precisely, we add a vertex $v_C$ into $A$ and place it very close to an arbitrary vertex of $C$ ``inside'' of $C$. We connect $v_C$ with all the vertices of $C$ by edges that closely follow the closed curve representing $C$ either from the left or from the right, and attach to their endpoints on $C$ from ``inside''; see Figure~\ref{obr_4_4_wheel}. We allow portions of these new edges to lie ``outside'' of $C$ only near self-crossings of $C$. In particular, in the neighborhoods of vertices of $C$, the new edges are always ``inside'' of $C$. Since no $C$-bridge is inner, all the new edges are even.

\begin{figure}
\begin{center}
 \ifpdf\includegraphics[scale=1]{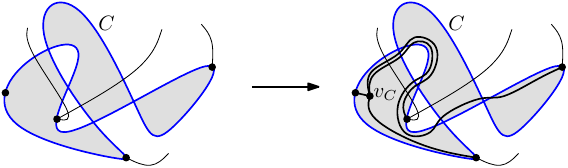}\fi
\end{center}
 \caption{Attaching a wheel to $C$.}
 \label{obr_4_4_wheel}
\end{figure}

\bigskip

Let $G''$ denote the graph obtained after processing all the cycles of $G'[A]$ and $G'[B]$.
Now we apply the strong Hanani--Tutte theorem to $G''$.
We further modify the resulting embedding in several steps so that in the end, the only vertices and edges of $G''$ not incident to the outer face of $G''[A]$ or $G''[B]$ are the vertices $v_C$ that form the centers of the wheels, and their incident edges. First, suppose that some of the wheels are embedded so that their central vertex $v_C$ is in the outer face of the wheel. Then the outer face is a triangle, say $v_Cuw$. We can then redraw the edge $uw$ along the path $uv_Cw$, without crossings, so that $v_C$ gets inside the wheel. We fix all the wheels in this way. Next, if some of the wheels contains another part of $G''$ in some of its inner faces, we flip the whole part over an edge of the wheel to its outer face, without crossings. See Figure~\ref{obr_4_5_flipping}. After finitely many flips, all the inner faces of the wheels will be empty.

\begin{figure}
\begin{center}
 \ifpdf\includegraphics[scale=1]{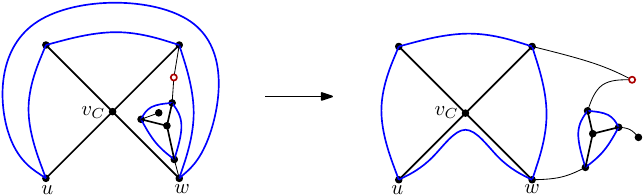}\fi
\end{center}
 \caption{Fixing the wheels and flipping everything else to the outer face of $G''[A]$. The circle represents a vertex in $B$.}
 \label{obr_4_5_flipping}
\end{figure}

After the modifications, $G''[A]$ is drawn in the outer face of $G''[B]$ and vice-versa. In the resulting embedding we delete all the vertices $v_C$ and contract the edges between the pairs of vertices $v', v''$ that were obtained by vertex-splits.

Thus, we obtain an embedding of $G'$ in which for every component $X_i$ of $G'[A]\cup G'[B]$, all vertices of $G'-X_i$ are drawn in the outer face of $X_i$. Now we insert the removed parts of $G$ back to $G'$, by copying the corresponding parts of the embeddings $\mathcal{D}(X_i)$ defined in the beginning of the proof. This is possible since we are placing the removed parts of $X_i$ inside faces bounded by simple cycles of $X_i$.
Hence, we obtain an embedding of $G$ in which for every component $X$ of $G[A]\cup G[B]$, all vertices of $G-X$ are drawn in the outer face of $X$.
Thus, Lemma~\ref{lemma:twoClusters} applies and that concludes the proof.

\section{Strong Hanani--Tutte for c-connected clustered graphs}
\label{section_c_connected}

\def\almost{1}
\def\almostt{0}

\begin{figure}
\begin{center}
 \ifpdf \includegraphics[scale=1]{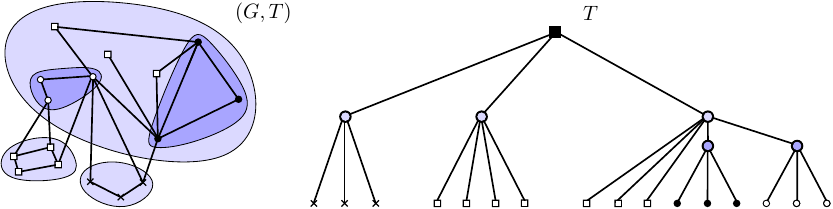} \fi
\end{center}
 \caption{A c-planar embedding of a c-connected clustered graph $(G,T)$ and the corresponding tree $T$.}
 \label{fig_5_1_c-connected1}
\end{figure}

Here we prove Theorem~\ref{thm:StrongClusterHTC}, using the ideas from the proof of Theorem~\ref{thm:StrongClusterHT}.

\if\almost\almostt
Let $(G,T)$ be an almost almost connected clustered graph with an independently even clustered drawing.
\else
Let $(G,T)$ be a c-connected clustered graph with an independently even clustered drawing. Our goal is to find a c-planar embedding of $(G,T)$; see Figure~\ref{fig_5_1_c-connected1}.
\fi
We proceed by induction on the number of clusters of $(G,T)$.
If the root cluster is the only cluster in $(G,T)$, the theorem follows directly from the strong Hanani--Tutte theorem applied to $G$. For the inductive step, we assume that $(G,T)$ has at least one non-root cluster.

A {\em minimal cluster\/} is a cluster that contains no other cluster of $(G,T)$.
\if\almost\almostt

A cluster $C$ is {\em connected\/} if $C$ induces a connected subgraph of $G$; otherwise $C$ is {\em disconnected}.
We distinguish two cases.

\paragraph{1) $(G,T)$ has a connected minimal cluster.}
Let $C$ be a connected minimal cluster of $(G,T)$.
\else
Let $V(\mu)$ be a minimal cluster of $(G,T)$.
\fi
Let $(G,T')$ be a clustered graph obtained
from $(G,T)$ by removing $\mu$ from $T$ and attaching all its children to its parent. Note that $(G,T')$ is still c-connected.
\if\almost\almostt
\jk{or almost almost connected.}
\fi

Starting from $(G,T')$, we process the connected subgraph $G[V(\mu)]$ analogously as the components of $G[A]$ in the proof of Theorem~\ref{thm:StrongClusterHT}, where we substitute $A=V(\mu)$ and $B=V(G)-V(\mu)$. By modifying $(G,T')$ we obtain a c-connected clustered graph $(G'',T'')$ with an independently even clustered drawing.
Now we apply the induction hypothesis and obtain a clustered embedding of $(G'',T'')$. Again, we modify this embedding so that all vertices of $V(G'')-V(\mu)$ are in the outer face of $G''[V(\mu)]$. Then we remove the wheels, contract the new edges and insert back the removed parts of $G[V(\mu)]$. Finally we draw a topological disc $\Delta(\mu)$ around the closure of the union of all interior faces of $G[V(\mu)]$. Since $G[V(\mu)]$ is connected, this last step is straightforward and results in a clustered embedding of $(G,T)$.

\if\almost\almostt
\jk{*** ted vlastne uz mame c-connected - uplne zadarmo! klidne napsat neco jako "Note that..."}

\paragraph{2) Every minimal cluster of $(G,T)$ is disconnected.}
Let $D$ be a minimal cluster of $(G,T)$. Since $(G,T)$ is almost almost connected, $D$ has no sibling clusters in $T$ (but it may have siblings that are vertices). Let $E$ be the parent of $D$. Let $\Delta_D$ and $\Delta_E$ be the two topological discs representing the clusters $D$ and $E$, respectively.  We process the connected components of $G[D]$ analogously to the proof of Theorem~\ref{thm:StrongClusterHT}, where we substitute $A=D$ and $B=V(G)-D$.

Next we process the connected components of $G[E\setminus D]$ as follows. By contracting $G[D]$ within the corresponding topological disc $\Delta_D$ to a single vertex $v_D$, and by contracting $G[V\setminus E]$ outside $\Delta_E$ to a single vertex $v_E$, we obtain an independently even drawing of a modified graph $G_{D,E}$. By the Hanani--Tutte theorem, $G_{D,E}$ has an embedding.

\fi

\section{Counterexample on three clusters}
\label{section_3clusters}

In this section we construct a family of even clustered drawings of flat clustered cycles on three and more clusters that are not clustered planar.
These examples imply that a straightforward generalization of the Hanani--Tutte theorem to graphs with three or more clusters is not possible.

Before giving the construction, we prove that there are no other ``minimal'' counterexamples to the Hanani--Tutte theorem for flat clustered cycles with three clusters, and more generally, flat clustered cycles whose clusters form a cycle structure. A reader interested only in the counterexample can immediately
proceed to Subsection~\ref{sub_counter} or directly to the study of Figure~\ref{obr_6_counter}.

Let $k\ge 3$. We say that a flat clustered graph $(G,T)$ with $k$ clusters is {\em cyclic-clustered\/} if there is a cyclic ordering of its clusters $(V_1, V_2, \dots, V_k)$ such that for $i\neq j$, $G$ has an edge between $V_i$ and $V_j$ if and only if $\lvert i-j\rvert \in \{1, k-1\}$; that is, if $V_i$ and $V_j$ are consecutive in the cyclic ordering. In this section we assume that $(G,T)$ is a cyclic-clustered graph with $k$ clusters. Clustered drawings of cyclic-clustered graphs with no edge-crossings outside the clusters have a simple structure.

\begin{observation}\label{obs_cyclic_clustered}
Let $\mathcal{D}$ be a clustered drawing of a cyclic-clustered graph $(G,T)$ with $k$ clusters on the sphere such that the edges do not cross outside the topological discs $\Delta_i$ representing the clusters $V_i$. Then we can draw disjoint simple curves  $\alpha_1, \beta_1, \alpha_2, \beta_2, \dots, \alpha_k,\beta_k$ such that both $\alpha_i$ and $\beta_i$ connect the boundaries of $\Delta_i$ and $\Delta_{i+1}$, do not intersect other discs $\Delta_j$, and the bounded region bounded by $\alpha_i$, $\beta_{i}$ and portions of the boundaries of $\Delta_i$ and $\Delta_{i+1}$ contains all portions of the edges between $V_i$ and $V_{i+1}$ that are outside of $\Delta_i$ and $\Delta_{i+1}$ (the indices are taken modulo $k$).
\end{observation}

\begin{proof}
The observation is obvious when there is exactly one edge between every pair of consecutive clusters. The general case follows easily by induction on the number of the inter-cluster edges.
\end{proof}

We note that if $(G,T)$ has only three clusters, then the conclusion of Observation~\ref{obs_cyclic_clustered} holds even if $(G,T)$ is not cyclic-clustered, that is, if there is a pair of clusters with no edge between them.

First we show that it is enough to consider clustered drawings in which the clusters are drawn
as cones bounded by a pair of rays emanating from the origin. We call such drawings \emph{radial}.

We call two clustered drawings of $(G,T)$ \emph{equivalent} if for every pair of independent edges $e$ and $f$, the number of their crossings has the same parity in both drawings. We call a clustered drawing \emph{weakly even} if every pair of edges between two disjoint pairs of clusters cross an even number of times. Clearly, every independently even drawing is also weakly even.

\begin{lemma}
\label{lemma_radial}
Given a weakly even clustered drawing $\mathcal{D}$ of a cyclic-clustered graph $(G,T)$,
there exists a radial clustered drawing of $(G,T)$ equivalent to $\mathcal{D}$.
\end{lemma}

\begin{proof}
Here we refer to the topological discs representing the clusters simply by ``clusters'', and denote them also by $V_i$.

If all the crossings in $\mathcal{D}$ are inside clusters, we can easily obtain a radial drawing of $(G,T)$ equivalent to $\mathcal{D}$ as follows. By Observation~\ref{obs_cyclic_clustered}, we can flip some edges so that the outer face intersects all the clusters. Then the complement of the union of the discs $\Delta_i$ and the curves $\alpha_i$ and $\beta_i$ from Observation~\ref{obs_cyclic_clustered} in the plane contains exactly one bounded and one unbounded component touching all the clusters. Therefore, we can continuously deform the plane and then expand the clusters to take the shape of the cones.

Suppose that there are crossings outside clusters in $\mathcal{D}$.
We show how to obtain an equivalent drawing that has all crossings inside clusters, in two phases.

In the first phase, we eliminate all crossings outside clusters as follows. We continuously deform every edge of $G$ between two consecutive clusters $V_i$ and $V_{i+1}$ (the indices are taken modulo $k$) into a narrow corridor between $V_i$ and $V_{i+1}$, keeping the interiors of $V_i$ and $V_{i+1}$ fixed except for a small neighborhood of their boundaries. See Figure~\ref{obr_6_1_weakly_even_to_radial}. Inside the corridor between $V_i$ and $V_{i+1}$, we want the portions of the edges to be noncrossing, but their order may be arbitrary. We may represent this deformation by the set $S( \mathcal{D}, \mathcal{D}')$ of edge-cluster switches (see Section~\ref{section_alg} for the definition) that were performed an odd number of times.

\begin{figure}
\begin{center}
 \ifpdf\includegraphics[scale=1]{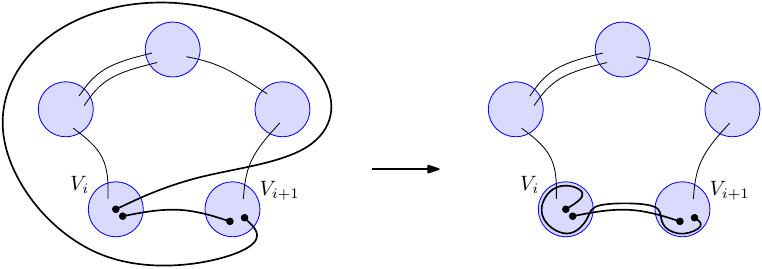}\fi
\end{center}
 \caption{Eliminating crossings outside clusters in a cyclic-clustered graph.}
 \label{obr_6_1_weakly_even_to_radial}
\end{figure}

Now we again use the fact that between every two consecutive clusters of the cyclic sequence $(V_1,\allowbreak V_2,\allowbreak \dots,\allowbreak V_k)$, there is at least one edge of $G$. Since no two edges cross outside clusters in $\mathcal{D}'$, both drawings $\mathcal{D}$ and $\mathcal{D}'$ are weakly even. Hence, if $S(\mathcal{D}, \mathcal{D}')$ contains an edge-cluster switch $(e,V_i)$ with a cluster $V_i$ that is disjoint with $e$, then $S(\mathcal{D}, \mathcal{D}')$ contains an edge-cluster switch of $e$ with every cluster disjoint with $e$. We call such an edge \emph{switched}.

In the second phase, we further transform $\mathcal{D}'$ into a drawing $\mathcal{D}''$
by deforming the edges only inside the clusters. For every switched edge $e$, we perform edge-vertex switches of $e$ with all vertices in the two clusters incident to $e$, except for the endpoints of $e$. Since performing an edge-vertex switch of $e$ with every vertex of $G$ not incident to $e$ has no effect on the parity of the number of crossings of $e$ with independent edges, the new drawing  $\mathcal{D}''$ is equivalent to $\mathcal{D}$.
\end{proof}

In the rest of this section we assume that $G$ is a cycle $C_n=v_1v_2\dots v_n$. For technical reasons, we define $v_{n+1}$ as $v_1$.
For $j\in [n]$, let $\varphi(v_j)$ denote the index of the cluster containing $v_j$, that is, $v_j\in V_{\varphi(v_j)}$.

For every edge $v_iv_{i+1}$ of $C_n$ we define $\mathrm{sign}(v_iv_{i+1}) \in \{-1, 0, 1\}$, as an element of $\mathbb{Z}$, so that $\mathrm{sign}(v_iv_{i+1}) \equiv \varphi(v_{i+1})-\varphi(v_{i}) \ (\text{mod } k)$.
Note that the sign is well defined since $(G,T)=(C_n,T)$ is cyclic-clustered and $k\ge 3$.
We then define the \emph{winding number} of $(C_n,T)$ as $\frac{1}{k}\sum_{i=1}^{n}\mathrm{sign}(v_iv_{i+1})$.
Note that in a radial clustered drawing of $(C_n,T)$ where the clusters $V_1, V_2, \dots V_n$ are drawn in a counter-clockwise order, our definition of the winding number of $(C_n,T)$ coincides with the standard winding number of the curve representing $C_n$ with respect to the origin.

We will show that if $(C_n,T)$ is a counterexample to the variant of the Hanani--Tutte theorem for flat cyclic-clustered graphs with $k$ clusters, then the winding number of $(C_n,T)$ is odd.

We say that $(C_n,T)$ is \emph{monotone} if
$\mathrm{sign}(v_1v_{2}) = \mathrm{sign}(v_2v_{3}) = \cdots = \mathrm{sign}(v_nv_{1}) \neq 0$.

In the following two lemmas we show how to reduce any even radial clustered drawing of $(C_n,T)$ to an even radial clustered drawing of a monotone cyclic-clustered cycle $(C_{n'},T')$, for some $n'\le n$, that has the same winding number as $(C_n,T)$.

We extend the notion of \emph{edge contraction} to flat clustered cycles as follows.
If $(G,T)$ is a clustered cycle and $e=uv$ is an edge of $G$ with both vertices $u,v$ in the same cluster $C$, then $(G,T)/e$ is the clustered multigraph obtained by contracting $e$ and keeping the vertex replacing $u$ and $v$ in the cluster $C$. The clustering of the rest of the vertices is left unchanged.
If $P=uwv$ is a path of length $2$ in $G$ such that $u$ and $v$ are in the same cluster $C$, then $(G,T)/P$ is the clustered multigraph obtained by contracting the edges $uw$ and $wv$ and keeping the vertex replacing $u$ and $v$ in the cluster $C$. Obviously, if $G=C_n$, then the contraction of an edge yields a cycle of length $n-1$. Similarly, the contraction of a path of length $2$ yields a cycle of length $n-2$.

\begin{lemma}
\label{lemma:contract}
Let $\mathcal{D}$ be an even radial clustered drawing of $(C_n,T)$.
Let $e$ be an edge in $C_n$ with both endpoints in the same cluster $V_i$.
Then $(C_n,T)/e$ has an even radial clustered drawing.
\end{lemma}

\begin{proof}
Since the edge $e$ is completely contained inside the disc representing the cluster $V_i$,
we can contract the curve representing $e$ in $\mathcal{D}$ towards one of its endpoints, dragging the edges incident to the other endpoint along. Since $e$ was even, this does not change the parity of the number of crossings between the edges of $G$.
\end{proof}

\begin{figure}
\begin{center}
 \ifpdf\includegraphics[scale=1]{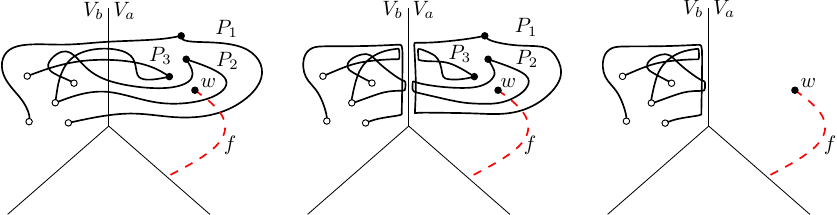} \fi
\end{center}
\caption{Illustration for the proof of Lemma~\ref{lemma:surgery}.
From left to right:
the successive stages of the redrawing operation eliminating paths $P_1, P_2$ and $P_3$. The edge $f$ cannot be present in the drawing, since it would violate its evenness.}
\label{fig_6_2_surgery}
\end{figure}

\begin{lemma}
\label{lemma:surgery}
Let $\mathcal{D}$ be an even radial clustered drawing of $(C_n,T)$. Let $V_a$ and $V_b$ be two adjacent clusters.
Let $P_1,\dots, P_m$ be all the paths of length $2$ in $C_n$ \\
whose middle vertices belong to $V_{a}$ and whose end vertices belong to $V_b$.
Then $(C_{n'},T')=(\dots((C_n,T)/P_1)/ \dots )/P_m$  has an even radial clustered drawing.
\end{lemma}

\begin{proof}
Refer to Figure~\ref{fig_6_2_surgery}.
By Lemma~\ref{lemma:contract}, we assume that no edge of $C_n$ has both vertices in the same cluster. At the end we can recover the contracted edges by decontractions.

The proof proceeds by the following surgery performed on $\mathcal{D}$. First we cut the paths $P_i$ at the ray $r$ separating the clusters $V_a$ and $V_b$, by removing a small neighborhood of the curves near $r$. Second, we reconnect the severed ends of every $P_i$ on both sides of $r$, by new curves drawn close to $r$. This operation splits every $P_i$ into two components. One of the components is a curve connecting the former end vertices of $P_i$, the other component is a closed curve containing the middle vertex of $P_i$. By removing the middle vertex of $P_i$, we replace each $P_i$ by a single edge $e_i$, still represented as the union of both components of $P_i$. Third, we remove the
closed curve of every $e_i$. Finally, we contract the remaining component of each $e_i$ towards one of the end vertices, as in Lemma~\ref{lemma:contract}.

We claim that the resulting drawing is even. It is easy to see that during the first and the second phase, the parity of the number of crossings between each pair of edges was preserved, if we consider the edge $e_i$ instead of each path $P_i$, and count the crossings on all components of every edge together. Now we show that the closed component of each $e_i$ crosses every other edge an even number of times. This is clearly true for every edge $e_j$ other than $e_i$, since only the closed component of $e_j$ can cross the closed component of $e_i$. Suppose that the closed component of $e_i$ crosses some other edge $f$ an odd number of times. Then $f$ intersects the region containing $V_a$, and so $f$ has one endpoint, $w$, in $V_a$. Since the other endpoint of $f$ is not in $V_a$, the vertex $w$ lies ``inside'' the closed component of $e_i$ (in the same sense as defined in Section~\ref{section_strong}). If some of the two edges incident to $w$ had the other endpoint outside $V_b$, it would cross $e_i$, and thus $P_i$, an odd number of times.
Therefore, both edges incident to $w$ are incident to both clusters $V_a$ and $V_b$. But every such pair was replaced by a single edge during the surgery; so there is no such $f$.
\end{proof}

\begin{theorem}
\label{thm:oddwind}
Let $(C_n,T)$ be a cyclic-clustered cycle that is not c-planar but has an even clustered drawing. Then the winding number of $(C_n,T)$ is odd and different from $1$ and $-1$.
\end{theorem}

\begin{proof}
Let $k\ge 3$ be the number of clusters of $(C_n,T)$. By Lemmas~\ref{lemma_radial}, \ref{lemma:contract} and~\ref{lemma:surgery}, we may assume that $(C_n,T)$ is monotone and that it has an even radial clustered drawing. In particular, the absolute value of the winding number of $(C_n,T)$ is equal to $n/k$.
Cortese et al.~\cite{CDibPP05_cycles} proved that a cyclic-clustered cycle is c-planar if and only if its winding number is $-1$, $0$ or $1$. This implies that $n\ge 2k$ if $(C_n,T)$ is not c-planar.

For every $i\in [k]$, we define a relation $<_i$ on $V_i$ as follows.
Refer to Figure~\ref{fig_6_3_oddwind}.
Let $u\in V_i$, and let $uu_{-}$ and $uu_{+}$ denote the two edges incident to $u$ so that $u_{-}\in V_{i-1}$ and $u_{+}\in V_{i+1}$ (the indices are taken modulo $k$). Let $(uu_{-})_i$ and $(uu_{+})_i$ denote the parts of $uu_{-}$ and $uu_{+}$, respectively,
contained inside the cone representing $V_i$. Let $r(uu_{-})$ and $r(uu_{+})$ denote the endpoint of $(uu_{-})_i$ and $(uu_{+})_i$, respectively, different from $u$. That is, $r(uu_{-})$ and $r(uu_{+})$ are on the boundary of the cone representing $V_i$.
 Let $\gamma(u)$ denote the closed curve obtained by concatenating $(uu_{-})_i$, $(uu_{+})_i$,
and the two line segments connecting $r(uu_{-})$ and $r(uu_{+})$, respectively, with the origin.
We say that vertices $u,v\in V_i$ are in the relation $u<_i v$ if $v$ is ``outside'' (in the same sense as defined in Section~\ref{section_strong})
of the curve $\gamma(u)$.

\begin{figure}
\begin{center}
 \ifpdf \includegraphics[scale=1]{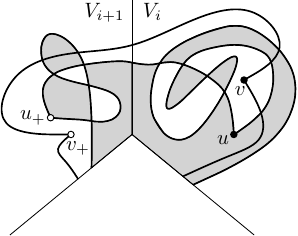} \fi
\end{center}
\caption{Illustration for the proof of Theorem~\ref{thm:oddwind}. The two pairs of vertices $u, v$, and $u_{+},v_{+}$ are in clusters $V_{i}$ and $V_{i+1}$, respectively.
The ``inside'' of the curves $\gamma(u)$ and $\gamma(u_{+})$ consists of the shaded regions. Thus, we have $u<_{i} v$ and $u_{+}<_{i+1} v_{+}$.}
\label{fig_6_3_oddwind}
\end{figure}

Let $v_{+}$ be the neighbor of $v$ in $V_{i+1}$, and let $v_{-}$ be the other neighbor of $v$.
The relations $<_i$ and $<_{i+1}$ satisfy the following properties.
\begin{enumerate}[(1)]
\item
  the relation $<_i$ is antisymmetric, that is, $(u<_i v) \Rightarrow \neg (v<_i u)$,
\item
  $u<_i v$ if and only if $u_{+}<_{i+1} v_{+}$.
\end{enumerate}

For part (1), we observe that $(vv_{-})_i$ and $(uu_{+})_i$ cross an even number of times. Suppose that $u<_i v$. Then $(vv_{-})_i$ and $(uu_{-})_i$ cross an odd number of times if and only if $r(vv_{-})$ is on $\gamma(u)$; equivalently, $r(vv_{-})$ is closer to the origin than $r(uu_{-})$. If also $v<_i u$, then $(vv_{-})_i$ and $(uu_{-})_i$ cross an odd number of times if and only if $r(uu_{-})$ is closer to the origin than $r(vv_{-})$; a contradiction.

For part (2), let $u_{++}$ be the neighbor of $u_{+}$ other than $u$. The claim follows from the fact that $vv_{+}$ crosses each of the curves $(uu_{-})_i, uu_{+}$ and $(u_{+}u_{++})_{i+1}$ evenly.

Recall that $C_n=v_1v_2\ldots v_n$. Let $i=\varphi(v_n)$ and suppose without loss of generality that $i+1=\varphi(v_1)$. Suppose that $n/k$ is even. Then both $v_n$ and $v_{n/2}$ are in $V_i$. By (2), we have $v_n<_i v_{n/2} \ \Leftrightarrow \ v_1<_{i+1} v_{n/2+1} \ \Leftrightarrow \ \cdots \ \Leftrightarrow \ v_{n/2}<_i v_n$, but this contradicts (1). Therefore, $n/k$ is odd.
\end{proof}

\paragraph{Remark.} 
We will see next that the relations $<_i$ are not necessarily transitive.
In fact, it is not hard to see that in every counterexample to the variant of the Hanani--Tutte theorem for cyclic-clustered cycles, no relation $<_i$ is transitive.

\subsection{Proof of Theorem~\ref{thm:counter}}
\label{sub_counter}



For simplicity of the description, we draw the graph on a cylinder, represented by a rectangle with the left and right side identified.

Let $r\ge 3$ be an odd integer and let $k\ge 3$. Our counterexample is a drawing of a monotone cyclic-clustered cycle with $kr$ vertices and $k$ clusters.
The corresponding curve consists of $kr+1$ periods of an appropriately scaled graph of the sinus function winding $r$ times around the cylinder, where the vertices mark the beginning of $kr$ of the periods.
We can describe the curve representing the cycle analytically as a height function $f(\alpha)=\sin\left( \frac {kr+1}{r}\alpha\right)$ on a vertical cylinder (whose axis is the $z$-axis) taking
the angle as the parameter. The vertices of the cycle are at points $\left(i\frac{2r}{kr+1}\pi, 0\right)$, where $i=0,\ldots, kr-1$, and
the clusters are separated by vertical lines at angles $\frac{2ri+1}{kr+1}\pi$, for $i=0,\ldots k-1$; see Figure~\ref{obr_6_counter}.
By the result of Cortese et al.~\cite{CDibPP05_cycles}, the cyclic-clustered cycle is not c-planar when $r>1$.

\begin{figure}
\begin{center}
 \ifpdf\includegraphics[scale=1]{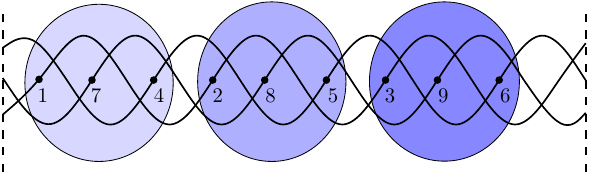}\fi
\end{center}
 \caption{A counterexample to the variant of the Hanani--Tutte theorem with parameters $k=3$ and $r=3$; the underlying graph is thus a cycle on $9$ vertices. The vertices are labeled by positive integers in the order of their appearance along the cycle.}
 \label{obr_6_counter}
\end{figure}

\section{Small faces}

\label{section_faces}

In this section we reprove a result of Di Battista and Frati~\cite{DiBF09_small_faces} that c-planarity can be decided in polynomial time for embedded flat clustered graphs whose every face is incident to at most five vertices.
In our proof, we reduce the problem to computing the largest size of a common independent set of two matroids. This can be done in polynomial time by the matroid intersection theorem~\cite{Edmonds01_submodular,Lawler75_matriod}.
See e.g.~\cite{Oxley11_book} for further references.

In this section,
we will use a shorthand notation $(G,T)$ instead of $(\mathcal{D}(G),T)$ for an embedded clustered graph.
Let $(G,T)$ be an embedded flat clustered graph where $G=(V,E)$.

Since contracting an edge with both endpoints in the same cluster does not affect $c$-planarity, we will assume that $(G,T)$ is an embedded clustered multigraph where every cluster induces an independent set. If $(G,T)$ is $c$-planar and contains a loop at $v$, then the whole interior of the loop must belong to the same cluster as $v$. Hence, either there is a vertex of another cluster inside the loop, in which case $(G,T)$ is not c-planar, or we may remove the loop and everything from its interior without affecting the c-planarity. The test and the transformation can be easily done in polynomial time. We will thus also assume that $(G,T)$ has no loops.

A \emph{saturator} of $(G,T)$ is a subset $S$ of ${V \choose 2}\setminus E$ such that every cluster of $(G\cup S,T)$ is connected and the edges of $S$ can be added to $(G,T)$ without crossings.

Let $S$ be a minimal saturator of $(G,T)$. Then each cluster in $(G\cup S,T)$ induces a spanning tree of the cluster, and so the boundary of each cluster can be drawn easily. We have thus the following simple fact.

\begin{observation}[{\cite{FCE95b_planarity}}]
\label{obs:matroid1}
An embedded flat clustered graph $(G,T)$ is c-planar if and only if $(G,T)$ has a saturator.
\end{observation}

In order to model our problem by matroids we need to avoid two noncrossing saturating edges in one face coming from
two different clusters, which might happen if the boundary of the face is not a
simple cycle.
To this end, we modify the multigraph further by a sequential merging of some pairs of vertices. Assuming that $u$ and $v$ are non-adjacent vertices incident to a common face $f$, \emph{merging} of $u$ and $v$ in $f$ consists in embedding a new edge $uv$ inside $f$ and then contracting it.

\begin{lemma}
\label{lemma:matroid}
Let $(G,T)$ be an embedded flat clustered multigraph all of whose faces are incident to at most five vertices. Suppose that $G$ has no loops and that every cluster of $(G,T)$ induces an independent set. Then there is an embedded flat clustered multigraph $(G',T)$ obtained from $(G,T)$ by merging vertices
such that
\begin{enumerate}[1)]
\item $(G,T)$ is c-planar if and only if $(G',T)$ is c-planar, and
\item if $(G',T)$ is c-planar then $(G',T)$ has a saturator $S$ whose edges can be embedded so that each face of $G'$ contains at most one edge of $S$.
\end{enumerate}
Moreover, finding $G'$ and verifying conditions 1) and 2) can be performed in linear time.
\end{lemma}

A \emph{saturating pair} of a face $f$ is a pair of vertices incident to $f$ and belonging to the same cluster. Thus, a cluster with $k$ vertices incident to $f$ has ${k \choose 2}$ saturating pairs in $f$. A \emph{saturating edge} of $f$ is a simple curve embedded in $f$ and connecting the vertices of some saturating pair of $f$.

\begin{proof}[Proof of Lemma~\ref{lemma:matroid}]
Clearly, once we find that $(G,T)$ is not c-planar we can choose $G'=G$.

A face of $(G,T)$ is \emph{bad} if it admits two noncrossing saturating edges, even from the same cluster. If no face of $(G,T)$ is bad, then the choice $G'=G$ satisfies both conditions of the lemma.

Assume that $(G,T)$ has at least one bad face $f$. We show that at least two vertices of $f$ can be merged so that the resulting embedded clustered multigraph is c-planar if and only if $(G,T)$ is c-planar. The lemma then follows by induction on the number of vertices.

Suppose that $f$ has only two saturating pairs, $\{u,v\}$ and $\{x,y\}$. In this case, $u$ and $v$ belong to a different cluster than $x$ and $y$. Since $f$ is bad, the pairs $\{u,v\}$ and $\{x,y\}$ can be joined by saturating edges $e(u,v)$  and $e(x,y)$, respectively, embedded in $f$ without crossings. Hence,
we can merge $u$ with $v$ along $e(u,v)$ while preserving the c-planarity.

If $f$ has more than two saturating pairs, there is a cluster $C$ that has at least three vertices incident to $f$. Let $C(f)$ be the set of these vertices. If all other clusters have at most one vertex incident to $f$, all saturating pairs of $f$ have vertices in $C(f)$. In this case, we can merge any pair of vertices of $C(f)$ while preserving the c-planarity.

In the remaining case, $f$ is incident to exactly five vertices, exactly three of them, $u,v$ and $w$, are in $C$, and the remaining two, $x$ and $y$, are in another cluster $D$. In this case, $f$ has four saturating pairs: $\{u,v\}$, $\{u,w\}$, $\{v,w\}$ and $\{x,y\}$. If $x$ and $y$ are in different components of the boundary of $f$, then it is possible to embed  saturating edges for all the four saturating pairs without crossings. We may thus merge $x$ with $y$ without affecting the c-planarity. For the rest of the proof we assume that $x$ and $y$ are in the same component of the boundary of $f$. In this case, every saturating edge $e$ joining $x$ with $y$ separates the face $f$ into two components. At least one of the components is incident to at least two vertices of $C(f)$, and so at least one saturating edge of the cluster $C$ can be embedded in $f$ while avoiding crossings with $e$. If at least two saturating edges of $C$ can be embedded in $f$ while avoiding crossings with $e$, we may merge $x$ with $y$ along $e$ without affecting c-planarity. Therefore, we also assume for the rest of the proof that for every saturating edge $e$ joining $x$ with $y$ in $f$, exactly one saturating pair of $C$ can be joined by a saturating edge embedded in $f$ without crossings with $e$. This implies that for every minimal saturator of $(G,T)$, at most two saturating pairs in total can be simultaneously joined by saturating edges embedded in $f$ without crossings.

If for some of the saturating pairs of $C$ in $f$, say, $\{u,v\}$, no saturating edge embedded in $f$ joining $x$ with $y$ separates $u$ and $v$, we can merge $u$ with $v$ without affecting c-planarity. We may thus assume that every pair of vertices in $C(f)$ can be separated by some saturating edge joining $x$ with $y$.

The boundary of $f$, denoted by $\partial f$, is a bipartite cactus forest with partitions $C(f)=\{u,v,w\}$ and $D(f)=\{x,y\}$. We call every connected component of $\mathbb{R}^2 \setminus \partial f$ other than $f$ an \emph{enclave}. Each enclave is bounded by a simple cycle, of length $2$ or $4$. Suppose that each enclave is bounded by a $2$-cycle. Since each of the $2$-cycles contains only one vertex of $C$, every saturating edge joining two vertices of $C(f)$ has to be embedded in $f$, and moreover, every minimal saturator of $(G,T)$ contains exactly two of the saturating pairs $\{u,v\}$, $\{u,w\}$, $\{v,w\}$, forming a spanning tree of the triangle $uvw$. Similarly, every minimal saturator of $(G,T)$ contains the pair $\{x,y\}$, and the saturating edge joining $x$ with $y$ must be embedded in $f$. By our assumptions, two of the three saturating edges in $f$ will cross, so in this case $(G,T)$ is not c-planar.

We are left with the case when one enclave is bounded by a $4$-cycle, say, $uxvy$. Clearly, there is at most one other enclave and it is bounded by a $2$-cycle. In total, there are three possibilities for the subgraph $\partial f$; see Figure~\ref{obr_7_enclave}. Every saturator of $(G,T)$ has to contain at least one of the two saturating pairs $\{u,w\}$, $\{v,w\}$, and the corresponding saturating edge must be embedded in $f$. Moreover, if there is a saturator containing $\{x,y\}$ and $\{v,w\}$ and no other saturating pair of $f$, then replacing $\{v,w\}$ with $\{u,w\}$ we also get a saturator, since saturating edges joining the pairs $\{x,y\}$ and $\{u,w\}$ can be simultaneously embedded in $f$ without crossings. 
Therefore, we can merge $u$ and $w$ while preserving c-planarity. This finishes the proof of the lemma.
\end{proof}

\begin{figure}
\begin{center}
 \ifpdf\includegraphics[scale=1]{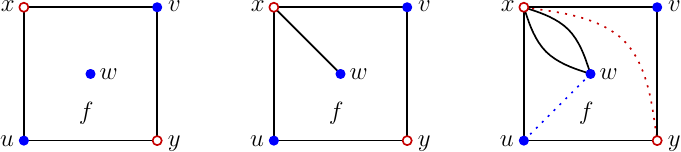}\fi
\end{center}
 \caption{Three cases of a bad face $f$ whose boundary contains a $4$-cycle. Saturating edges joining the pairs $\{x,y\}$ and $\{u,w\}$ are drawn in the third case. The vertices $u$ and $w$ can be merged without affecting the c-planarity.}
 \label{obr_7_enclave}
\end{figure}

\subsection{Proof of Theorem~\ref{thm:matroid}}

We start with the embedded multigraph $(G',T')$ obtained in Lemma~\ref{lemma:matroid}.
By Observation~\ref{obs:matroid1} and Lemma~\ref{lemma:matroid}, it is enough to decide whether $(G',T')$ has a minimal saturator.

In order to test the existence of a saturator
we define two matroids for which we will use the matroid intersection algorithm.
The ground set of each matroid is a set $\overline{E'}$ of saturating edges of $(G',T')$ defined as the disjoint union $\bigcup_f E_f$, over all faces of $G'$, where $E_f$ is a set containing one saturating edge for each saturating pair of $f$. By the proof of Lemma~\ref{lemma:matroid}, no face $f$ is bad, so every set $E_f$ has at most two saturating edges. Moreover, if $|E_f|=2$, then the two saturating edges in $E_f$ cross and belong to different clusters.

The first matroid, $M_1$, is the direct sum of graphic matroids constructed for each cluster as follows.
Denote the clusters of $(G',T')$ by $C_i$, $i=1,\ldots, k$. Let $G_i$ be the multigraph induced by $C_i$ in $\overline{G'}=(V,\overline{E'})$. If $G_i$ is not connected, $(G',T')$ has no saturator. We thus further assume that $G_i$ is connected.
The ground set of the graphic matroid $M(G_i)$ is the edge set of $G_i$. Since $G_i$ is connected, the rank of $M(G_i)$ is the number of vertices of $G_i$ minus one.
Since the matroids $M(G_i)$, $i=1,\ldots, k$, are pairwise disjoint, their direct sum, $M_1$, is also a matroid and its rank is the sum of the ranks of the matroids $M(G_i)$.

The second matroid, $M_2$, is a partition matroid defined as follows.
A subset of $\overline{E'}$ is independent in $M_2$ if it has at most one edge in every face of $G'$.

Let $M$ be the intersection of $M_1$ and $M_2$. If $M$ has an independent set of size equal to the rank of $M_1$, then $(G',T')$ has a saturator that has at most one edge
inside each face. Thus, $(G',T')$ is c-planar by Observation~\ref{obs:matroid1}, and that in turn implies by Lemma~\ref{lemma:matroid} that $(G,T)$ is c-planar as well.
On the other hand, if $(G,T)$, and hence $(G',T')$, is c-planar, then $(G',T')$ has a minimal saturator $S$ that
has at most one edge inside each face by Lemma~\ref{lemma:matroid}. Thus, $S$ witnesses the fact that $M$ has an independent set of size equal to the rank of $M_1$.
Hence, $(G',T')$ is c-planar if and only if $M$ has an independent set of size equal to the rank of $M_1$, and this can be tested by the matroid intersection algorithm.

\section{Concluding remarks}

\label{section_epilogue}
Let $G_T$ be the simple graph obtained from $(G,T)$ by contracting the clusters and deleting the loops and multiple edges.
By the construction in Section~\ref{section_3clusters} we cannot hope for a fully general variant of the Hanani--Tutte theorem for $(G,T)$ when $G_T$ contains a cycle.

A simple modification of the construction provides a counterexample also for the case when $G_T$ is a tree with at least one vertex of degree greater than 2; see Figure~\ref{obr_8_counter_tree}. This disproves our conjecture from the conference version of this paper~\cite{FKMP14_ht}.

\begin{figure}
\begin{center}
 \ifpdf\includegraphics[scale=1]{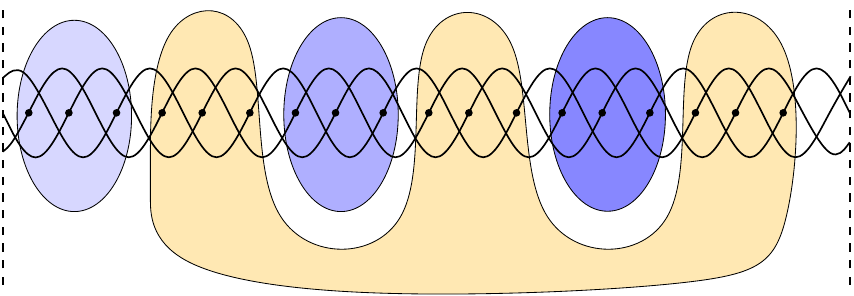}\fi
\end{center}
 \caption{A counterexample with $G_T=K_{1,3}$.}
 \label{obr_8_counter_tree}
\end{figure}

Therefore, the only open case for flat clustered graphs is the case when $G_T$ is a collection of paths. We conjecture that the strong Hanani--Tutte theorem holds in this case.

\begin{conjecture_}
\label{conj:conj}
If $G_T$ is a path and $(G,T)$ admits an independently even clustered drawing then $(G,T)$ is c-planar.
\end{conjecture_}

A variant of Conjecture~\ref{conj:conj} for non-flat two-level clustered graphs in which the clusters on the bottom level form a path and one additional cluster contains all interior clusters of the path would provide a polynomial-time algorithm for
c-planarity testing for strip clustered graphs, which is an open problem stated in~\cite{ADDF13+}.

Our proof from Section~\ref{section_faces} fails if the graph has hexagonal faces.
We wonder if this difficulty can be overcome or rather could lead to NP-hardness.

\section*{Acknowledgements}
We are grateful to the anonymous referee and to the numerous anonymous reviewers of early versions for many valuable comments.




\end{document}